\documentclass[lettersize,journal]{IEEEtran}

\usepackage{amsmath,amsthm,amssymb,amsfonts,verbatim}
\usepackage[colorlinks,
            linkcolor=blue,
            anchorcolor=blue,
            citecolor=blue
            ]{hyperref}
\usepackage[hmargin=1.2in,vmargin=1.2in]{geometry}
\usepackage{array}
\usepackage{multirow}
\usepackage{algcompatible}
\usepackage{algorithm}
\usepackage{makecell}
\usepackage{hhline}

\title{Explicit Construction of $q$-ary $2$-deletion Correcting Codes with Low Redundancy}
\author{Shu Liu\thanks{Shu Liu is with the National Key Laboratory of Science and Technology on Communications, University of Electronic Science and Technology of China, Chengdu 611731, China (email: shuliu@uestc.edu.cn).},
Ivan Tjuawinata\thanks{Ivan Tjuawinata is with the Strategic Centre for Research on Privacy-Preserving Technologies and Systems, Nanyang Technological University, Singapore 637553 (email: ivan.tjuawinata@ntu.edu.sg).},
Chaoping Xing\thanks{Chaoping Xing is with the School of Electronic Information and Electric Engineering, Shanghai Jiao Tong University, Shanghai, China (email: xingcp@sjtu.edu.cn).}}
\date{}

\newtheorem{lemma}{Lemma}[section]
\newtheorem{theorem}[lemma]{Theorem}
\newtheorem{cor}[lemma]{Corollary}

\newtheorem{claim}[lemma]{Claim}

\newtheorem{rem}[lemma]{Remark}
\newtheorem{defn}{Definition}

\theoremstyle{remark}

\renewcommand{\epsilon}{\varepsilon}
\renewcommand{\le}{\leqslant}
\renewcommand{\ge}{\geqslant}

\newcommand{\vnote}[1]{}

\def \balpha {{\boldsymbol\alpha}}

\def \Xi {{X^{[i]}}}

\usepackage{etoolbox} 
\newtoggle{comments}
\newtoggle{anonymous}
\newtoggle{short}
\newtoggle{future}
\toggletrue{anonymous}
\toggletrue{comments}
\toggletrue{short}
\togglefalse{future} 

\usepackage{todonotes}
\iftoggle{future}{
  \newcommand{\future}[1]{\todo[color=blue!25,inline]{\textsf{TODO Future:} #1}}
}{
  \newcommand{\future}[2][]{}
}

\iftoggle{comments}{%
  \newcommand{\xcp}[2][]{\todo[color=yellow!50,#1]{\textsf{XCP:} #2}}
}{
  \newcommand{\xcp}[2][]{}
}
\iftoggle{comments}{%
  \newcommand{\ivan}[2][]{\todo[color=blue!50,#1]{\textsf{ivan:} #2}}
}{
  \newcommand{\ivan}[2][]{}
}

\begin{document}
\onecolumn
\maketitle

\begin{abstract} We consider the problem of  efficient construction of $q$-ary $2$-deletion correcting codes with low redundancy. We show that our construction requires less redundancy than any existing efficiently encodable  $q$-ary $2$-deletion correcting codes. Precisely speaking, we present an explicit construction of a $q$-ary $2$-deletion correcting code with  redundancy $5\log n+10\log\log n + 3\log q+O(1).$
Using a minor modification to the original construction, we obtain an efficiently encodable $q$-ary $2$-deletion code that is efficiently list-decodable. Similarly, we show that our construction of list-decodable code requires a smaller redundancy compared to any existing list-decodable codes.

To obtain our sketches, we  transform a $q$-ary codeword to a binary string which can then be used as an input to the underlying base binary sketch. This is then complemented with additional $q$-ary sketches that the original $q$-ary codeword is required to satisfy. In other words, we build our codes via a binary $2$-deletion code as a black-box. Finally we utilize the binary $2$-deletion code proposed by Guruswami and H{\aa}stad to our construction to obtain the main result of this paper.
\end{abstract}
\section{Introduction}\label{sec:intro}

We consider the construction of a $q$-ary $2$-deletion correcting code. A $q$-ary $2$-deletion correcting code is a set of  strings of a fixed length $n$ over an alphabet $\Sigma_q$ of size $q, \mathcal{C}\subseteq \Sigma_q^n$ such that for any two distinct codewords $\mathbf{c}_1,\mathbf{c}_2\in \mathcal{C},$ they do not have any common subsequence of length at least $n-2.$ A natural question is whether we can efficiently construct such $q$-ary $2$-deletion correcting code with as large size as possible. Equivalently, this can be formulated as the problem of designing efficient construction of $q$-ary $2$-deletion correcting code with minimum amount of redundancy. Here we define the redundancy of the code to be the additional bits required during the encoding to facilitate the error correction. More specifically, we define the redundancy of a $q$-ary code $\mathcal{C}\subset \Sigma_q^n$ of size $M$ to be $n\log_2 q - \log_2 M.$  Denote by $R_D(q,n,2)$ (and $R_S(q,n,2)$, respectively) to be the smallest redundancy of  $q$-ary $2$-deletion (and substitution, respectively) error  correcting  codes of length $n$.  Throughout this work, for simplicity, the notation $\log(\cdot)$ is used to denote the logarithm base $2.$ It was shown in \cite{Lev02} that
\[2\log n + 2\log q + o(\log q\log n)\le R_D(q,n,2)\le 4\log n + 2\log q + o(\log q\log n).\]

It is then interesting to see if an efficient construction of $q$-ary $2$-deletion correcting code with redundancy approaching the optimal redundancy provided above can be designed.

\subsection{Related Work}\label{sec:intro-RW}
There have been numerous studies on low redundancy Hamming metric codes with fixed Hamming distance (see \cite{YD04} for instance). In particular, we have $R_S(q,n,2)=2\log n+o(\log n)$ for $q=2,3,4$ and $2\log n+o(\log n)\le R_S(q,n,2)\le \frac73\log n+o(\log n)$ for $q>4$.

There have been various works on the construction of $t$-deletion codes. The study of $t$-deletion code construction was first considered in the construction of the binary Varshamov-Tenengolts codes (VT codes for short) \cite{VT65,Lev65} which was extended to $q$-ary codes in \cite{Ten84}. VT codes, which can correct one deletion error, was then further extended to the binary Helberg code \cite{HF02}, which can correct multiple deletions. Helberg code was further extended to be defined over alphabets of size $q$ in \cite{LN16}. All these constructions are based on the concept of sketch where codes are constructed by collecting all solutions to a pre-defined system of equations. Asymptotically, the redundancy of the resulting code equals the size of the corresponding sketch.

Although Hamming metric codes with logarithmic redundancy have been studied for many years, study on deletion correcting codes with logarithmic redundancy was initiated in \cite{BGZ18} recently, which proposed binary codes capable of correcting multiple deletions. More specifically, they construct binary $k$-deletion codes with redundancy $O(k^2\log k \log n).$ There have been some improvements on this construction, especially in the case when $k=2.$ In \cite{SRB18}, a binary $2$-deletion codes with sketch of size $7\log n+6$ bits was proposed while a binary $8\log n+ O(\log \log n)$ was proposed in \cite{GS19}. A more general improvement on binary $k$-deletion codes was proposed with redundancy $8k\log n + o(\log N)$ in \cite{SB19} as long as $k=o(\sqrt{\log \log n}).$ A binary $2$-deletion codes was more recently constructed with the help of both sketch and regularity assumption in \cite{GH21} to produce a binary code of redundancy $4\log n + O(\log \log n).$ In contrast to the previous constructions, the work of \cite{GH21} provides a smaller sketch size by requiring that the codewords satisfy a regularity assumption. In order to enforce such assumption while maintaining the rate of the code, an explicit encoding process to regular strings is proposed. Recently, there have also been some works on the construction of codes capable of correcting both deletion and substitution \cite{SWWY20, SPCH22}. In their work \cite{SPCH22}, a binary $t$-deletion correcting code has also been proposed with redundancy $(4t-1)\log n + o(\log n).$

There have also been some works on the construction of $q$-ary $2$-deletion codes following the line of work previously discussed. In \cite{Zhou20}, a $q$-ary $2$-deletion codes was constructed with existing binary $2$-deletion codes as the underlying base. In the proposed construction, an additional $3\log n +O_q(\log \log n)$ bits of sketch is required. The work has also considered a simpler construction which produces a code with redundancy $6\log n +O_q(\log \log n)$ that can list decode any $2$ deletions with list size $2.$ In \cite{SC22}, they consider the case when $q>2$ is even and also constant with respect to $n.$ In such case, they proposed a $q$-ary $2$-deletion codes with redundancy $5\log n +O(\log q \log \log n).$ In \cite{SPCH22}, a $q$-ary $t$-deletion $s$-substitution correcting systematic code was proposed where $q$ is at most logarithmic with respect to the code length. In such construction, when $s$ is set to be $0,$ the construction provides a $q$-ary $t$-deletion correcting code of redundancy $4t\log_q n+o(\log_q n).$

\subsection{Our Result}\label{sec:intro-Res}
In this work, we consider a general construction of $q$-ary sketch which is based on sketches of binary $2$-deletion codes. Based on this construction, we show that the resulting construction can be used to construct $q$-ary $2$-deletion codes with lower redundancy compared to the previously best known result. We further note that the encoding and decoding algorithm of such codes are linear with respect to the code length and the square of the logarithm of the alphabet size (Theorem~\ref{MT:UD1}). Such construction is also used to construct $q$-ary deletion codes that can list decode against any $2$ deletions with smaller redundancy compared to the previously best known result. Again, we show that the encoding and decoding complexities of such code is linear with respect to the code length and the square of the logarithm of the alphabet size(Theorem~\ref{MT:LD1}).

In order to take advantage of the construction of the binary $2$-deletion codes proposed in \cite{GH21}, we further modify our construction to allow for the regularity assumption to again be satisfied. This results in $q$-ary deletion codes with better redundancy compared to the general construction we have previously discussed (Theorem~\ref{thm:ConsFromGH21}). Here we note that although the redundancy of the constructed codes is reduced, it comes with the increase of encoding and decoding complexity. More specifically, the resulting codes have encoding complexity of $\mathrm{poly}(n)$ while the efficiency of the decoding complexity relies on the efficiency of the inversion of the encoding process. The main result of this paper is given below.

\begin{theorem}\label{thm:main}
For any positive integer $q>2,$ there exists an explicit and efficiently encodable $q$-ary $2$ deletion correcting code of length $n$ with redundancy $5\log n+10\log\log n + 3\log q+O(1).$

Furthermore,  there exists an explicit and efficiently encodable $q$-ary deletion code of length $n$ with redundancy $4\log n+10\log\log n+2\log q+O(1)$ that can be  list decoded against $2$ deletions with list size of at most $2.$

\end{theorem}

\begin{rem}{\rm As we mentioned in the earlier subsection, the lowest redundancy of $q$-ary $2$ deletion correcting codes of length $n$ with explicit construction in literature is given by  $5\log n +O(\log q \log \log n)$ in \cite{SC22}. It is easy to see that our redundancy is smaller than that of \cite{SC22}.
}
\end{rem}
\subsection{Our Technique}\label{sec:intro-Tech}
In order to construct our sketch, we define a function that transforms our $q$-ary codewords to a binary string which can then be used as an input to the underlying base binary sketch. This is then complemented with additional $q$-ary sketches that the original $q$-ary codeword is required to satisfy. Having these two steps of sketches, the decoding algorithm can also then be designed in two separate steps. First, based on a received $q$-ary string, the same transformation function can be applied to generate the corresponding binary string, which can then be shown to be obtainable by deleting $2$ elements from the binary string obtained from the sent codeword. This, together with the binary sketches, allows us to recover the original binary string corresponding to the sent codeword. We can then utilize such binary string as well as the $q$-ary sketches to recover the original codeword. An observation to be made in such decoding algorithm is that the last step of the decoding algorithm, which requires $2\log q$-bits of sketches can be removed and this results in a list decoding algorithm with list size of at most $2.$ This provides us with the list-decoding variant of our construction with smaller sketch size.

In order to utilize the binary $2$-deletion code proposed in \cite{GH21}, we require that our message is encoded to a $q$-ary string with some regularity assumption. Similar to \cite{GH21}, we show that such encoding algorithm can be designed explicitly and the encoding process has a sufficiently large domain size to maintain the redundancy size.

\subsection{Organization}\label{sec:intro-org}
The remainder of the paper is organized as follow. In Section~\ref{sec:prelim}, we review some basic definitions and existing results that will be useful in the discussion in the following sections. Section \ref{sec:gencons1} focuses on a general construction of $q$-ary $2$-deletion correcting code from existing binary $2$-deletion correcting code with specific form. Such construction is then realized using previously proposed binary $2$ deletion correcting codes, the resulting redundancy is considered and compared in Section~\ref{sec:GenConsExandComp}. In the same section, we also consider a variant of such construction to allow for list-decodable $2$-deletion codes with smaller redundancy. Lastly, a specific construction based on the binary $2$-deletion correcting code proposed in \cite{GH21} is considered and analyzed in Section~\ref{sec:ConsFromGH21}.

\section{Preliminary}\label{sec:prelim}

Let $q$ be a prime number, $n$ be positive integers and for any positive integer $m,$ let $\Sigma_m=\{0,1,\cdots, m-1\}\subseteq \mathbb{Z}.$ 
Given a string $\mathbf{x}=(x_1,\cdots, x_n)$ and a positive string $m\le n,$ a string $\mathbf{y}=(y_1,\cdots, y_m)$ is said to be a subsequence of $\mathbf{x}$ if there exists $1\le i_1<i_2<\cdots<i_m\le n$ such that $y_j=x_{i_j}$ for $j=1,\cdots, m.$ Furthermore, $\mathbf{y}$ is said to be a substring of $\mathbf{x}$ if it is a subsequence of $\mathbf{x}$ and for any $j=2,\cdots, m, i_j=i_{j-1}+1.$

\begin{defn}\label{def:qto2}
\begin{enumerate}
\item Let $\balpha=(\alpha_1,\cdots, \alpha_n)\in \Sigma_q^n.$ Define $\mathbf{a}=(a_1,\cdots, a_n)\in \Sigma_2^n$ such that for $i=1,\cdots, n,$
\[a_i=\left\{
\begin{array}{cc}
1, &\mathrm{~if~}i=1\mathrm{~or~}\alpha_i\ge\alpha_{i-1}\\
0, &\mathrm{~otherwise}
\end{array}
\right..\]
We denote such map by $\varphi_n:\Sigma_q^n\rightarrow \Sigma_2^n$ where it applies to any length $n.$ We observe that $\varphi_n$ can be calculated in $O(n)$ time.
\item Suppose that $\balpha\in \Sigma_q^n$ is associated to $\mathbf{a}=(a_1,\cdots, a_n)\in \Sigma_2^n$ by the definition above. Recall that a run in $\mathbf{a}$ is defined to be a substring of $\mathbf{a}$ with the same symbol. More specifically, $\mathbf{a}$ has a run from position $i$ to position $j$ for some $1\le i\le j\le n$ if $a_i=a_{i+1}=\cdots=a_j\in \Sigma_2,$ either $i=1$ or $a_{i-1}=1-a_i$ and either $j=n$ or $a_{j+1}=1-a_j.$ We can define $\mathbf{a}$ as a sequence of runs. More specifically, assuming that $\mathbf{a}$ has $r$ runs, there exists $b\in \Sigma_2$ such that $\mathbf{a}=(1^{\ell_1}\|0^{\ell_2}\|\cdots\| b^{\ell_r})$ where $\ell_i>0$ and $\sum_{i=1}^r \ell_i=n$ and
\[b=\left\{
\begin{array}{cc}
1,\mathrm{~if~}i\mathrm{~is~odd}\\
0,\mathrm{~otherwise}
\end{array}
\right..\]
For non-negative $i,$ we define
\[j_i=\left\{
\begin{array}{cc}
0,\mathrm{~if~}i=0\\
\ell_i+j_{i-1},\mathrm{~otherwise}
\end{array}
\right..\]

For $i=1,\cdots, r,$ we define the $i$-th run of $\balpha$ associated to $\mathbf{a}$ by the vector $(\alpha_{j_{i-1}+1},\cdots, \alpha_{j_i}).$ It is easy to see that for any $i=1,\cdots, r,$ the $i$-th run of $\balpha$ is either decreasing or non-decreasing. Lastly, we define $\beta_i=\alpha_{j_{i-1}+1}+\alpha_{j_{i-1}+2}+\cdots+\alpha_{j_i}.$ In the following, given $\balpha\in \Sigma_q^n$ and $\mathbf{a}=\varphi_n(\balpha),$ we denote by $r_\balpha$ the number of runs of $\mathbf{a}.$ Furthermore, for $i=1,\cdots, r,$ we call $\beta_i$ by the $i$-th \emph{run sums} of $\balpha.$ We again note that such $\beta_i$ can be computed in $O(n)$ given the value of $\varphi_n(\balpha).$

\end{enumerate}
\end{defn}

Next, we provide an observation on $\varphi$ with respect to a deletion operation.
\begin{claim}\label{claim:varphidel}
Let $m$ be a positive integer, $\balpha\in \Sigma_q^m$ and $\mathbf{a}=\varphi_m(\balpha)\in \Sigma_2^m.$ Suppose that $\balpha'\in \Sigma_q^{m-1}$ is a subsequence of $\balpha$ by deleting one of its entries and suppose that $\mathbf{a}'=\varphi_{m-1}(\balpha)\in\Sigma_2^{m-1}.$ Then $\mathbf{a}'$ is a subsequence of $\mathbf{a}.$
\end{claim}
\begin{proof}
Let $i\in[m]$ such that $\alpha_i$ is deleted from $\balpha$ to obtain $\balpha'.$ First suppose that $i=1.$ In such case, we note that for $j=2,\cdots, m-1, a_j'=a_{j+1}$ and $a_1'=1=a_1.$ Hence it is easy to see that $\mathbf{a}'$ is a subsequence of $\mathbf{a}.$ Secondly, suppose that $i=m.$ In this case, it is easy to see that for $j=1,\cdots, m-1, a_j'=a_j.$ Hence, it is again clear to see that $\mathbf{a}'$ is a subsequence of $\mathbf{a}.$ Lastly, suppose that $2\le i\le m-1.$ Note that in this case, for $j\in[m-1],$
\[a_j'=\left\{
\begin{array}{cc}
a_j,&\mathrm{~if~}j\le i-1\\
1,&\mathrm{~if~}j=i\mathrm{~and~} \alpha_{i+1}\ge \alpha_{i-1}\\
0,&\mathrm{~if~}j=i\mathrm{~and~} \alpha_{i+1}<\alpha_{i-1}\\
a_{j+1},&\mathrm{~if~} j\ge i
\end{array}
\right.\]
Hence to prove the claim, it is sufficient to show that $a_i'$ is a subsequence of $(a_i,a_{i+1}).$ Since $a_i',a_i,a_{i+1}\in \Sigma_2,$ if $a_i\neq a_{i+1}, a_i'$ must be equal to either $a_i$ or $a_{i+1}.$ Hence it must be a subsequence of $(a_i,a_{i+1}).$ On the other hand, if $a_i=a_{i+1},$ we must have either $\alpha_{i+1}\ge \alpha_i\ge \alpha_{i-1},$ which implies that $a_i=a_{i+1}=a_i'=1,$ or $\alpha_{i+1}<\alpha_i<\alpha_{i-1},$ which implies that $a_i=a_{i+1}=a_i'=0.$ So in any case, we again obtain that $a_i'$ is a subsequence of $(a_i,a_{i+1}).$ This conclude our proof.

\end{proof}

Note that by applying Claim~\ref{claim:varphidel} multiple times, we obtain the following corollary.
\begin{cor}\label{cor:varphigen}
Let $m<n$ be a positive integer, $\balpha\in \Sigma_q^n$ and $\mathbf{a}=\varphi_n(\balpha).$ Suppose that $\balpha'\in \Sigma_q^{m}$ is a subsequence of $\balpha$ by deleting $n-m$ of its entries and suppose that $\mathbf{a}'=\varphi_m(\balpha').$ Then $\mathbf{a}'$ is a subsequence of $\mathbf{a}.$
\end{cor}

We will briefly discuss some results in~\cite{GH21} that will be useful in our discussion. First, we discuss the definition of regular binary string.

\begin{defn}\label{def:RegStr} Let $d$ be a positive integer. We say that a string $\mathbf{x}\in \Sigma_2^n$ is $d$-regular if each substring of $\mathbf{x}$ of length at least $d \log_2 n$ contains both $00$ and $11$ as substrings.

We can further extend this definition to a $q$-ary string. More specifically, we say that $\mathbf{x}\in \Sigma_q^n$ is $d$-regular if its associated binary string $\mathbf{a}\in \Sigma_2^n$ is $d$-regular. This corresponds to the requirement that any substring $\mathbf{x}'$ of $\mathbf{x}$ of length at least $d\log_2 n$ must contain a substring of length three $(y_1,y_2,y_3)$ such that $y_1>y_2>y_3.$ Furthermore, we also require that it contains a substring $(z_1,z_2,z_3)$ such that $z_1\le z_2\le z_3$ unless $\mathbf{x}'$ contains the first entry of $\mathbf{x}.$ In the case that $\mathbf{x}'$ contains the first entry of $\mathbf{x}$ and it does not contain any substring $(z_1,z_2,z_3)$ such that $z_1\le z_2\le z_3,$ then we require that $x_1\le x_2$ to allow for the first two entries of $\mathbf{a}$ to all be $1.$
\end{defn}


%

In this work, we consider constructions of codes with the use of sketch functions. We note that by the Pigeonhole Principle, the construction of a $q$-ary sketch function with output size of $\ell$-bits that can be used to correct up to $2$ deletions implies the existence of a $q$-ary $2$-deletion codes with redundancy of at least $\ell.$ When explicitness of the code is required, the redundancy required is larger. More specifically, similar to the discussion in \cite[Lemma $2.1$]{GH21} and \cite[Lemma $1$]{Zhou20}, we can obtain the following relation.

\begin{lemma}\label{lem:Sketch-Code}
Fix an integer $k\geqq 1.$ Let $s:\Sigma_q^n\rightarrow \Sigma_2^{c_1 \log n + c_2\log q + O(1)}$ be an efficiently computable function such that for any $\mathbf{x}\in \Sigma_q^n,$ given $s(\mathbf{x})$ and any $\mathbf{y}\in \Sigma_q^{n-2},$ a subsequence of $\mathbf{x},$ the original $\mathbf{x}$ can be efficiently and uniquely recovered. Then there exists an efficiently encodable and decodable code $\mathcal{C}\subseteq \Sigma_q^N$ of size $q^n$ with length $N\leq n+c_1\log n + c_2\log q + O(\log \log n)$ such that $\mathcal{C}$ is a $2$-deletion correcting codes.
\end{lemma}

Intuitively, this can be achieved by the use of the given $q$-ary sketch $s$ as well as a binary sketch $\bar{s}$ with output length $O(\log(n'))$ given any input length $n',$ which can be satisfied by sketches proposed in \cite{BGZ18, GS19, SRB18}. The encoding is then done by appending the message $\mathbf{m}$ by the sketch $s(\mathbf{m})$ followed by repeating the sketch $\bar{s}(s(\mathbf{m}))$ three times. The decoding can then be done by first recovering $\bar{s}(s(\mathbf{m})),$ followed by using $\bar{s}(s(\mathbf{m}))$ to recover $s(\mathbf{m}).$ Lastly, having $s(\mathbf{m}), \mathbf{m}$ can then be recovered.

We can then see that the encoding complexity equals the complexity to calculate $s$ given an input of length $n$ and to calculate of $\bar{s}$ with input of length $c_1\log n+c_2\log q+O(1).$ Note that since $\bar{s}$ is assumed to be efficiently computable, its complexity is polynomial in the input length. Hence in our case, the complexity is polynomial in $O(\log n+\log q).$ Hence the encoding complexity asymptotically equals the complexity of calculating $s$ with input length $n,$ which is again assumed to be efficient.

Similarly, the decoding process contains three steps. The first step includes the decoding process of repetition code with length $O(\log (\log n+\log q))$ which is then sub-linear with respect to $n$ and $q.$ The second step is the decoding process of $\bar{s}$ with input length $O(\log n+\log q).$ Since the decoding process of $\bar{s}$ is polynomial in its length by assumption, its complexity in our application is hence $\mathrm{poly}(\log n+\log q).$  The last step is the decoding process of $s$ with input length $n.$ Hence, the dominating factor of the decoding complexity is the complexity of the decoding process of $s$ with input length $n.$

Due to the relation described in Lemma~\ref{lem:Sketch-Code}, in the remainder of this work, we focus on the construction of such sketch function and its output length, which can be directly seen as the asymptotic redundancy of the resulting explicit $q$-ary $2$-deletion correcting code.

\section{General Construction}\label{sec:gencons1}

Here we construct a $q$-ary sketch function that can correct up to $2$ deletions from an existing binary sketch function capable of correcting up to $2$ deletions. Suppose that there exists a sketch function $\bar{s}$ such that for any $\mathbf{x}\in \Sigma_2^n,$ given $\bar{s}(\mathbf{x})$ and any subsequence $\mathbf{y}\in \Sigma_2^{n-2}$  of $\mathbf{x},$ the original $\mathbf{x}$ can be uniquely determined.

Let $\tilde{s}$ be a function that takes a $q$-ary string $\mathbf{x}$ as an input, computes $\bar{\mathbf{x}}=\varphi(\mathbf{x})$ and outputs $\bar{s}(\bar{\mathbf{x}}).$ We further define the following three sketches $s^{(1)},s^{(2)},s^{(3)}:\Sigma_q^n\rightarrow \mathbb{Z}$ such that for $\balpha=(\alpha_1,\cdots, \alpha_n)\in \mathbb{Z}_q^n$ with $r_{\balpha}$ runs and run sums $\beta_1,\cdots,\beta_{r_\balpha},$
\begin{enumerate}
\item $s^{(1)}(\balpha)=\sum_{i=1}^n \alpha_i,$
\item $s^{(2)}(\balpha)=\sum_{i=1}^n \alpha_i^2,$ and
\item $s^{(3)}(\balpha)=\sum_{i=1}^{r_\balpha}i\beta_i.$
\end{enumerate}

It is easy to see that for the three sketches defined above, they can be computed in $O(n).$

\begin{lemma}\label{lem:gencons}
For any $\balpha\in \Sigma_q^n,$ given the values of $s^{(1)}(
\balpha)\pmod{q},s^{(2)}(\balpha)\pmod{q},$ $s^{(3)}(\balpha)\pmod{nq},$ $\tilde{s}(\balpha),$ and any  subsequence $\balpha'=(\alpha'_1,\cdots, \alpha'_{n-2})\in \Sigma_q^{n-2}$  of $\balpha,$ the original $\balpha$ can be uniquely recovered.
\end{lemma}
\begin{proof}
Let $\mathbf{a}'=(a_1',\cdots, a'_{n-2})=\varphi(\balpha').$ By Corollary~\ref{cor:varphigen}, we have that $\mathbf{a}'$ is a subsequence of $\mathbf{a}.$ So we have $\mathbf{a}'\in \Sigma_2^{n-2},$ a subsequence of $\mathbf{a}$ obtained by deleting two entries. Furthermore, we also have $\tilde{s}(\balpha)=\bar{s}(\mathbf{a}).$ Hence by assumption, given these, we can uniquely determine $\mathbf{a}=(1^{\ell_1}\|0^{\ell_2}\|\cdots\|1^{\ell_{r_\balpha}})$ where $b=1$ if $r_\balpha$ is odd and $b=0$ otherwise. This means that we also obtain the runs of $\mathbf{a}$ as well as their lengths $\ell_j$ where $\ell_i>0$ and $\sum_{i=1}^{r_{\balpha}} \ell_i=n.$ In particular, we also obtain the values of $1\le k_1\le k_2\le r_\balpha$ that determines the runs where the two entries of $\balpha$ are deleted from.

Suppose that the deleted symbols from $\balpha$ in the $k_1$-th and $k_2$-th runs are $v_1,v_2\in\Sigma_q$ respectively. By assumption, we know the values of $s^{(i)}\equiv s^{(i)}(\balpha)\pmod{q}$ for $i=1,2.$ Given $\balpha',$ we can also compute $\hat{s}^{(i)}\equiv s^{(i)}(\balpha')\pmod{q}.$ Then for $i=1,2,$ we have $v_1^i+v_2^i\equiv \Delta_i \pmod{q}$ where $\Delta_i\triangleq s^{(i)}-\hat{s}^{(i)}\pmod{q}.$ Hence we have the following system of equations
\begin{equation}\label{eq:findvi}
\left\{
\begin{array}{ccc}
v_1+v_2&\equiv&\Delta_1\pmod{q}\\
v_1^2+v_2^2&\equiv&\Delta_2\pmod{q}
\end{array}
\right.\end{equation}
Squaring the first equivalence and subtracting the second from it yields $2v_1v_2\equiv \Delta_1^2-\Delta_2\pmod{q}.$ Since $q$ is an odd prime, $2^{-1}\pmod{q}$ exists and define $\Delta_0\in \mathbb{Z}_q$ such that $\Delta_0\equiv 2^{-1}\left(\Delta_1^2-\Delta_2\right)\pmod{q}.$ Hence treating the equations over $\mathbb{Z}_q,$ we have $v_1+v_2=\Delta_1$ and $v_1v_2=\Delta_0.$ This implies that $v_1$ and $v_2$ are roots of the quadratic equation $x^2-\Delta_1 x + \Delta_0=0.$ Define $S=\{a\in \mathbb{Z}_q:\exists b\in \mathbb{Z}_q, b^2\equiv a\pmod{q}\}$ be the set of quadratic residues modulo $q.$ Note that by definition, given that $b^2\equiv a\pmod{q},$ we must also have $(q-b)^2\equiv a\pmod{q}.$ Hence for any $a\in S,$ there exists a unique $b\in\left\{0,1,\cdots, \frac{q-1}{2}\right\}\subseteq \mathbb{Z}_q$ such that $b^2\equiv a\pmod{q}.$ We denote such $b,$ if it exists, by $\sqrt{a}.$ Now note that under the condition that $b\triangleq\sqrt{\Delta_1^2-4\Delta_0}$ exists, we have that $v_1,v_2$ are different elements of the set $\left\{2^{-1}\left(\Delta_1+b\right)\pmod{q},2^{-1}\left(\Delta-b\right)\pmod{q}\right\}.$ However, note that by definitions of $\Delta_0$ and $\Delta_1,$ we have $\Delta_1^2-4\Delta_0=(v_1-v_2)^2.$ Hence $b$ is the unique element in $\{(v_1-v_2)\pmod{q},(v_2-v_1)\pmod{q}\}\cap\left\{0,1,\cdots, \frac{q-1}{2}\right\}.$ This shows that provided that the square root function $\sqrt{\cdot}$ can be efficiently computed, we can compute $\{v_1,v_2\}$ efficiently.

Note that although we have learned the set $\{v_1,v_2\},$ which we denote by $\theta_1$ and $\theta_2,$ unless $\theta_1=\theta_2,$ we need to determine whether $(v_1,v_2)=(\theta_1,\theta_2)$ or $(v_1,v_2)=(\theta_2,\theta_1).$ Here with the help of $s^{(3)}(\balpha)$ we prove that given $\{\theta_1,\theta_2\}$ and $k_1,k_2,$ the run numbers where the two deleted entries are from, we can uniquely recover $\balpha.$ Note that if $k_1=k_2,$ since $\theta_1\neq \theta_2,$ there is a unique way to place $\theta_1$ and $\theta_2$ in the $k_1$-th run to satisfy the run definition. Hence, it remains to show that with the possible help of $s^{(3)}(\balpha),$ we can uniquely recover $\balpha$ from $\{\theta_1,\theta_2\}$ and $k_1,k_2$ given that $k_1\neq k_2$ and $\theta_1\neq \theta_2.$ Without loss of generality, when seen as positive integers, assume that $\theta_1<\theta_2$ and $k_1<k_2.$ Note that here $\theta_2-\theta_1<q$ and $k_2-k_1<n.$ Hence $0<(\theta_2-\theta_1)(k_2-k_1)<nq.$ We divide the case depending whether the number of runs of $\balpha$ and $\balpha^\prime$ are different.

\begin{itemize}
\item Case $1:$ Both the $k_1$-th run and $k_2$-th run also exist in $\balpha^\prime.$ This implies that after the re-insertion, both $k_1$-th and $k_2$-th run contain elements from $\balpha^\prime.$ For $i=1,2,$ let $\balpha^{(i)}$ be the word obtained from $\balpha^\prime$ by adding $\theta_1$ to the $k_i$-th run of $\balpha^\prime$ and $\theta$ to the $k_{3-i}$-th run of $\balpha^\prime.$ Then $s^{(3)}(\balpha^{(1)})\equiv B+k_1\theta_1+k_2\theta_2\pmod{nq}$ while $s^{(3)}(\balpha^{(2)})\equiv B+k_1\theta_2+k_1\theta_1\pmod{nq}.$ Hence $s^{(3)}(\balpha^{(1)})\equiv s^{(3)}(\balpha^{(2)})\pmod{nq}$ if and only if $(k_2-k_1)(\theta_2-\theta_1)\equiv 0\pmod{nq}.$ However, this is impossible since $0<(k_2-k_1)(\theta_2-\theta_1)<nq.$ Hence in this case, there can only be one of the possibilities that can have $s^{(3)}(\balpha^{(i)})=s^{(3)}(\balpha).$
\item Case $2:$ Both the $k_1$-th run and $k_2$-th run did not exist in $\balpha^\prime.$ Note that in this case, after the re-insertion, both $k_1$-th and $k_2$-th run are of length $1.$ Note that because of this, the actual position of the deleted elements are determined, say at $i_1<i_2.$ Hence we have that the first symbol is inserted between $\alpha^\prime_{i_1-1}$ and $\alpha^\prime_{i_1}$ while the second symbol is inserted between $\alpha^\prime_{i_2-2}$ and $\alpha^\prime_{i_2-1}.$ Note that if at least one of $\balpha^{(1)}$ or $\balpha^{(2)}$ does not satisfy the required value of $\mathbf{a},$ we can directly eliminate such value from the possible value of $\balpha.$ So we assume that both $\mathbf{a}=\varphi(\balpha^{(1)})=\varphi(\balpha^{(2)}).$
Then $s^{(3)}(\balpha^{(i)})\equiv B+k_1\theta_1+2 \sum_{j=i_1}^{i_2-2}\alpha_i'+k_2\theta_{3-i}+4\sum_{j=i_2-1}^{n-2}\alpha_i'\pmod{nq}$ for $i=1,2.$ So we again have that $s^{(3)}(\balpha^{(1)})\equiv s^{(3)}(\balpha^{(2)})\pmod{nq}$ if and only if $(k_2-k_1)(\theta_2-\theta_1)\equiv 0\pmod{nq} $ which is again impossible since $0<(k_2-k_1)(\theta_2-\theta_1)<nq.$
\item Case $3:$ Suppose that the $k_1$-th run did not exist in $\balpha^\prime$ while the $k_2$-th run did. We again note that if there exists $i$ such that $\varphi(\balpha^{(i)})\neq \mathbf{a},$ we can directly eliminate such possibility for $\balpha.$ So we consider the case when $\mathbf{a}=\varphi(\balpha^{(1)})=\varphi(\balpha^{(2)}).$ Using a similar argument as before, $s^{(3)}(\balpha^{(i)})\equiv B+k_1\theta_1+2\sum_{j= i_1}^{n-2}\alpha_j^\prime + k_2\theta_{3-i}\pmod{nq}.$ So $s^{(3)}(\balpha^{(1)})\equiv s^{(3)}(\balpha^{(2)})\pmod{nq}$ if and only if $(k_2-k_1)(\theta_2-\theta_1)\equiv 0\pmod{nq}$ which cannot happen by our assumption of $k_1,k_2,\theta_1$ and $\theta_2.$
\item Case $4:$ Lastly, suppose that $k_1$-th run existed in $\balpha^\prime$ while the $k_2$-th run did not. Using a similar argument as before, $s^{(3)}(\balpha^{(i)})\equiv B+k_1\theta_1+ k_2\theta_{3-i}+2\sum_{j= i_2-1}^{n-2}\alpha_j^\prime\pmod{nq}.$ So $s^{(3)}(\balpha^{(1)})\equiv s^{(3)}(\balpha^{(2)})\pmod{nq}$ if and only if $(k_2-k_1)(\theta_2-\theta_1)\equiv 0\pmod{nq}$ which cannot happen by our assumption of $k_1,k_2,\theta_1$ and $\theta_2$ which concludes our proof.
\end{itemize}

%
%
%
%
%
%
%

\end{proof}

It is easy to see that by definition, the sketch we use, which consists of $\bar{s},s^{(1)}(\balpha)\pmod{q},s^{(2)}(\balpha)\pmod{q}$ and $s^{(3)}(\balpha)\pmod{nq}$ can be computed in $T_{\bar{s},E}(n)+O(n)$ where $T_{\bar{s},E}(n)$ is the complexity of calculating $\bar{s}.$ Hence our sketch is efficiently computable provided that $\bar{s}$ is efficiently computable.

Next, we consider the efficiency of the decoding algorithm implied by our proof above. After the decoding of $\mathbf{a}$ from $\mathbf{a}',$ which uses the decoding algorithm of the underlying binary sketch, the values of $v_1$ and $v_2$ are determined by solving the System \eqref{eq:findvi}, whose complexity is asymptotically equal to the complexity of calculating the square root function over $\mathbb{Z}_q.$ Hence, if a square root algorithm, for example, the Tonelli-Shanks algorithm is used \cite{Ton91,Sha73}, its complexity is $O(\log^2 q).$ Having $k_1,k_2,v_1$ and $v_2,$ the last step of the decoding algorithm is to identify whether to insert $v_1$ or $v_2$ at the $k_1$-th or $k_2$-th run. This can be done by considering both possibilities and calculate the value of $s^{(3)}$ for both possibilities. Hence, it is easy to see that such calculation can be done in time linear with respect to $n.$ So the overall decoding algorithm takes time $T_{\bar{s},D}(n)+O(\log^2 q + n)$ where $T_{\bar{s},D}(n)$ is the decoding complexity of the underlying binary sketch.

 This results in our the following efficiently computable $q$-ary sketch function.
\begin{theorem}\label{thm:gencons}
Let $\bar{s}:\Sigma_2^n\rightarrow \Sigma_2^{\ell(n)}$ be an efficiently computable sketch function such that for any $\mathbf{a}\in \Sigma_2^n,$ given $\bar{s}(\mathbf{a})$ and any subsequence $\mathbf{a}'\in \Sigma_2^{n-2}$ of $\mathbf{a},$ the original $\mathbf{a}$ can be uniquely and efficiently recovered. We further assume that $\bar{s}$ can be computed in time $T_{\bar{s},E}(n)$ while the error correction takes $T_{\bar{s},D}(n).$ Then for any positive integer $q>2,$ there exists an explicit sketch $s:\Sigma_q^n\rightarrow \Sigma_2^{\log n+3\log q}$ that can be efficiently computed such that for any $\balpha\in \Sigma_q^n,$ given $\bar{s}(\varphi(\balpha)), s(\balpha)$ and any subsequence $\balpha^\prime\in \Sigma_q^{n-2}$ of $\balpha,$ the original $\balpha$ can be uniquely and efficiently recovered. In other words, the total sketch size is $\ell(n)+\log n+3\log q$-bits. Furthermore, the sketch calculation and the error correction complexity are $T_{\bar{s},E}(n)+O(n)$ and $T_{\bar{s},D}(n)+O(\log^2 q+n)$ respectively.
\end{theorem}

\begin{rem}\label{rem:genconscode}{\rm
By Lemma~\ref{lem:Sketch-Code}, such sketch implies the existence of a $q$-ary $2$-deletion code of length $N\le n+\ell(n)+\log n+3\log(q)+O(\log \log n)$ with encoding complexity $T_{\bar{s},E}(n)+O(n)$ and decoding complexity $T_{\bar{s},D}(n)+O(\log^2 q+n).$
}\end{rem}
\section{Construction Realization and Comparison with Existing Construction}\label{sec:GenConsExandComp}
In this section, we use the construction shown in Theorem \ref{thm:gencons} with some existing binary sketch functions to obtain the actual output length of the resulting construction. Then, we compare such construction with existing $q$-ary sketch functions. Furthermore, we also consider the use of our construction and its variant to construct $q$-ary sketch functions capable of list decode up to two deletions.

\subsection{Instantiation using Existing Binary Two-Deletion Sketch Function and its Comparison}

In the following we consider the length of the resulting sketch function based on various existing two-deletion codes proposed in \cite{GS19,SRB18,SPCH22}. We note that the code and sketch described in \cite{SRB18} has $O(n)$ encoding and decoding algorithms, the ones described in \cite{SPCH22} has $O(n^5)$ and $O(n^3)$ encoding and decoding algorithms while the one proposed in \cite{GS19} does not come with the encoding and decoding complexities. Hence, we no longer have the efficiency assumption for the construction based on \cite{GS19}.

\begin{cor}\label{cor:GenConsRed}
\begin{itemize}
\item Let $\mathcal{C}_0^{(1)}\subseteq\Sigma_2^n$ be the binary two-deletion correcting code proposed in \cite{GS19} which has a sketch of size $8\log n+O(\log \log n)$ bits. Then we can compute a sketch function $s_0^{(1)}$ with total output size $9\log n+O(\log \log n)+3\log q$ bits that can be used to correct up to $2$ deletions.

\item Let $\mathcal{C}_0^{(2)}\subseteq \Sigma_2^n$ be the binary two-deletion correcting codes proposed in \cite[Construction $2$]{SPCH22} which has a sketch of size $7\log n+o(\log n)$ bits. \\
Then we can efficiently compute a sketch function $s_0^{(2)}$ with total output size $8\log n + 3\log q + o(\log n)$ bits that can be used to efficiently correct up to $2$ deletions. Moreover, the sketch can be computed in $O(n^5)$ while the error correction can be done in $O(\log^2 q + n^3).$

\item Let $\mathcal{C}_0^{(3)}\subseteq \Sigma_2^n$ be the binary two-deletion correcting codes proposed in \cite{SRB18} which has a sketch of size $7\log n+6$ bits. \\
Then we can efficiently compute a sketch function $s_0^{(3)}$ with total output size $8\log n + 3\log q + \log 6$ bits that can be used to efficiently correct up to $2$ deletions. Moreover, the sketch can be computed in $O(n)$ while the error correction can be done in $O(\log^2 q + n).$
\end{itemize}
\end{cor}

Note that the second point of Corollary \ref{cor:GenConsRed} yields an efficiently encodable and decodable $q$-ary $2$-deletion code, which will be summarized in the following theorem.

\begin{theorem}\label{MT:UD1}
There is an explicit $q$-ary $2$-deletion correcting code $\mathcal{C}\subseteq \Sigma_q^N$ of size $q^n$ with $N\leq n+ 8\log n + 3\log q + O(\log \log n)$ with redundancy $8\log n+3\log q + O(\log \log n)$ with encoding complexity $O(n)$ and decoding complexity $O(n+\log^2 q).$
\end{theorem}

 \begin{rem}\label{rem:ConsComp}{\rm
We note that the construction discussed above cannot be directly used with the binary sketch proposed in \cite{GH21}. This is because such construction requires the binary string to be regular, which is not a requirement that can be expressed as a sketch. So based on Corollary \ref{cor:GenConsRed}, the smallest sketch from our initial construction requires output length of $8\log n + 3\log q + 6$ bits.

Note that when $q>2$ is even and constant with respect to $n,$ a construction of a $q$-ary sketch function that can be used to correct up to $2$ deletion with output $5\log n+(16\log q+10)\log \log n$ bits was proposed in \cite{SC22}. It is easy to see that the regime of $q$ considered in our work is different where we require $q$ to be an odd prime. Furthermore, in general, our best sketch requires larger output compared to the one provided in \cite{SC22}. However, our construction is still efficient even when $q$ is no longer constant with respect to $n.$ Furthermore, when $q$ is sufficiently large, the output length of our proposed sketch function is smaller compared to the extension of the construction provided in \cite{SC22}.

A similar construction of $q$-ary $2$-deletion sketch function was proposed in \cite{Zhou20}. Given an underlying binary sketch capable to correct $2$ deletions with output length $\ell$ bits, their proposed construction provides a $q$-ary sketch capable to correct $2$ deletions with output length $\ell+3\log n+O_q(\log \log n)$ bits.  It is easy to see that our resulting sketch has a smaller output length of $\ell+\log n + 3\log q$ bits.

Compared to the $q$-ary $2$-deletion correcting code implied by \cite[Construction 3]{SPCH22}, we note that it has the same asymptotic redundancy. However, our resulting construction has a better encoding and decoding complexity, which is $O(n)$ encoding and $O(n)$ decoding complexities instead of $O(n^5)$ encoding and $O(n^3)$ decoding complexities.
}
 \end{rem}

 \subsection{Construction of Sketch Function for List Decoding from Lemma~\ref{lem:gencons} and its Simplification}
Note that a $q$-ary sketch function that can be used to list decode against $2$-deletion can be constructed using a similar approach as Lemma~\ref{lem:gencons}. More specifically, if the underlying binary sketch function $s_0$ is list-decodable against $2$-deletion, our resulting sketch can also be used to list decode against $2$-deletion with the same list size. Alternatively, we note that without the sketch $s^{(3)}(\balpha),$ the decoding algorithm discussed in Lemma~\ref{lem:gencons} yields a list of size at most $2.$ Hence, if the underlying binary sketch function $s_0$ can be used to correct up to $2$ deletions, this gives a $q$-ary sketch that can list decode against $2$-deletion with list size of at most $2.$ On the other hand, if the underlying binary sketch $s_0$ can be used to list-decode against $2$-deletion errors with list size $\mathcal{L},$ our construction produces a $q$-ary sketch that can list decode against $2$-deletion with list size of at most $2\mathcal{L}.$ We summarize this discussion in the following Corollary.

\begin{cor}\label{cor:GenConsLD}
Let $\bar{s}_U:\Sigma_2^n\rightarrow \Sigma_2^{\ell(n)}$ be a sketch function that can be used to  decode up to $2$-deletion errors. Furthermore, let $\bar{s}_L:\Sigma_2^n\rightarrow \Sigma_2^{\ell'(n)}$ be a binary sketch function that can be used to list decode up to $2$-deletions with list size $\mathcal{L}$ for some positive integer $\mathcal{L}.$ We assume that the complexity of computing $\bar{s}_U, \bar{s}_L,$ the error correction using $\bar{s}_U$ and the list decoding using $\bar{s}_L$ are $T_{\bar{s}_U,E}(n),T_{\bar{s}_U,E}(n), T_{\bar{s}_U,D}(n)$ and $T_{\bar{s}_L,D}(n)$ respectively. Then for any $q>2,$ there exists:
\begin{enumerate}
\item A $q$-ary sketch function $s$ with output length $\ell(n)+2\log q$ bits that can be used to list-decode $2$ deletions with list size of at most $2.$ Furthermore, the complexities of the sketch computation and the list decoding algorithm are $T_{\bar{s}_U,E}(n)+O(n)$ and $T_{\bar{s}_U,D}(n)+O(\log^2 q)$ respectively.
\item A $q$-ary sketch function $s$ with output length $\ell'(n)+3\log q+\log n$ bits that can be used to list-decode $2$ deletions with list size of at most $2.$ Furthermore, the complexities of the sketch computation and the list decoding algorithm are $T_{\bar{s}_L,E}(n)+O(n)$ and $T_{\bar{s}_L,D}(n)+O(\log^2 q+n)$ respectively.
\item  A $q$-ary sketch function $s$ with output length $\ell'(n)+2\log q$ bits that can be used to list-decode $2$ deletions with list size of at most $4.$ Furthermore, the complexities of the sketch computation and the list decoding algorithm are $T_{\bar{s}_L,E}(n)+O(n)$ and $T_{\bar{s}_L,D}(n)+O(\log^2 q)$ respectively.
\end{enumerate}
\end{cor}

Here we provide instantiations based on existing constructions of binary deletion sketches in \cite{SRB18,GH21}. We note that the complexities of computing the list-decodable binary sketch proposed in \cite{GH21} and the list-decoding algorithm are both $O(n).$

\begin{cor}\label{cor:GenConsLDRed}
\begin{itemize}
\item Let $s_U$ be the binary sketch function proposed in \cite{SRB18} with output length $7\log n+\log 6$ bits. Then the $q$-ary sketch function $s^{(0)}$ generated following the construction discussed in the first point of Corollary \ref{cor:GenConsLD} can be used to efficiently list decode against any $2$ deletions with list size of at most $2$ and sketch length $7\log n + 2\log q + \log 6.$ Furthermore, the complexities of the sketch computation and the list decoding algorithm are $O(n)$ and $O(\log^2 q+n)$ respectively.
\item Let $s_L$ be a binary sketch function proposed in \cite{GH21} that can be used to efficiently list decode any two deletions with list of size at most $2$ and sketch length $3\log n+O(\log \log n).$ Then the sketch $s^{(1)}$ obtained by following the construction discussed in the second point of Corollary \ref{cor:GenConsLD} can be used to efficiently list decode any $2$ deletions with list size of at most $2$ and sketch length $4\log n + 3\log q +O(\log \log n).$ The complexities of the sketch computation and the list decoding algorithm are $O(n)$ and $O(\log^2 q+n)$ respectively.

Furthermore, the sketch $s^{(2)}$ obtained by following the construction discussed in the last point of Corollary \ref{cor:GenConsLD} can be used to efficiently list decode any $2$ deletions with list size of at most $4$ and output length $3\log n + 2\log q + O(\log \log n).$ The complexities of the sketch computation and the list decoding algorithm are $O(n)$ and $O(\log^2 q+n)$ respectively.
\end{itemize}
\end{cor}

\begin{rem}\label{rem:GenConsLDComp}{\rm
The work in \cite{Zhou20} also proposed a $q$-ary sketch function that can be used to list decode any $2$ deletions with list size of at most $2$ and sketch length $6\log n+O_q(\log \log n)$ bits. Recall that our construction provides a $q$-ary sketch that can be used to list decode against $2$ deletions with list size of at most $2$ with sketch length $4\log n + 3\log q  + O(\log \log n)$ bits.  It is then easy to see that as long as $q^2<n^3,$ our construction provides a smaller sketch.
}\end{rem}

The sketches constructed in Corollary \ref{cor:GenConsLDRed} yields a $q$-ary deletion code that can be efficiently encoded and list decoded against any $2$-deletions. Such result is summarized in the following theorem.

\begin{theorem}\label{MT:LD1}
There is an explicit $q$-ary deletion correcting code $\mathcal{C}\subseteq \Sigma_q^N$ of size $q^n$ with $N\leq n+ 4\log n + 3\log q + O(\log \log n)$ with redundancy $4\log n+3\log q + O(\log \log n)$ with encoding complexity $O(n)$ which is list decodable against any $2$-deletions with list size $2$ and decoding complexity $O(n+\log^2 q).$

Furthermore, there is an explicit $q$-ary deletion correcting code $\mathcal{C}\subseteq \Sigma_q^N$ of size $q^n$ with $N\leq n+ 3\log n + 2\log q + O(\log \log n)$ with redundancy $3\log n+2\log q + O(\log \log n)$ with encoding complexity $O(n)$ which is list decodable against any $2$-deletions with list size $4$ and decoding complexity $O(n+\log^2 q).$
\end{theorem}

\section{Construction of $q$-ary $2$-Deletion Correcting Code from \cite{GH21}}\label{sec:ConsFromGH21}

We first note that the general construction discussed in Lemma~\ref{lem:gencons} cannot be directly used if we use the binary $2$-deletion correcting sketch function $s^{(0)}$ proposed in \cite{GH21} as the underlying binary sketch. This is because the proposed construction requires the binary string considered to be $d$-regular for some constant $d.$ Note that in the construction in \cite{GH21}, this is done by providing a specific efficient encoding $\Pi$ of messages to regular strings, which is shown to have sufficiently large domain. Note that in our case, we cannot encode $\varphi(\balpha)$ using $\Pi$ to obtain $\mathbf{a}$ since we can no longer guarantee that the two deletions in $\balpha$ correspond to two deletions in $\mathbf{a}.$ Hence, a modification to the construction in Lemma~\ref{lem:gencons} is required where we further require that $\balpha$ itself is $d$-regular. This yields the following lemma

\begin{lemma}\label{lem:GSCons}
Let $\bar{s}$ be the sketch function used in \cite[Theorem $5.9$]{GH21}. Then for any $d$-regular $\balpha\in \Sigma_q^n,$ given $s^{(1)}(\balpha)\pmod{q}, s^{(2)}(\balpha)\pmod{q}, s^{(3)}(\balpha)\pmod{nq}, \tilde{s}(\balpha)$ and any subsequence $\balpha^\prime\in \Sigma_q^{n-2}$ of $\balpha,$ the original $\balpha$ can be uniquely and efficiently recovered.
\end{lemma}

This implies that if an efficient encoding function to regular $q$-ary string exists, the existence of an efficiently encodable $q$-ary $2$-deletion correcting code can then be guaranteed. In the following, we consider such encoding function. For any positive integer $m,$ define the set $S_m=\{\balpha\in \Sigma_q^m:\balpha\mathrm{~is~}d\mathrm{-regular}\}.$ We further define $S_m^{(0)}=\{\balpha\in \Sigma_q^m: \nexists i\in\{1,\cdots, m-2\}:\alpha_i>\alpha_{i+1}>\alpha_{i+2}\}, F_m^{(0)}=|S_m^{(0)}|,S_m^{(1)}=\{\balpha\in \Sigma_q^m: \nexists i\in\{1,\cdots, m-2\}:\alpha_i\le\alpha_{i+1}\le\alpha_{i+2}\}$ and $F_m^{(1)}=|S_m^{(1)}|.$

\begin{lemma}\label{lem:Fm}
For any $m\ge 3$ and $q>2, F_m^{(1)}<F_m^{(0)}<(0.99 q)^m.$
\end{lemma}
 \begin{proof}
 Note that the first inequality is clear since $S_{m}^{(1)}\subsetneq T_m^{(1)}\triangleq\{\balpha\in \Sigma_q^m:\nexists i\in\{1,\cdots, m-2\}:\alpha_i<\alpha_{i+1}<\alpha_{i+2}\}$ and $|T_m^{(1)}|=F_m^{(0)}.$

 Now, we note that $F_3^{(0)}=q^3-\binom{q}{3}$ and it can be verified that $F_3^{(0)}<(0.99q)^3$ for any $q>2.$ Next we consider $F_4^{(0)}.$ Note that for a $q$-ary string of length $4$ not to be considered in $S_4^{(0)},$ it must have a strictly decreasing substring of length $3$ either in its first $3$ or last $3$ elements. Hence, there are $2q\binom{q}{3}-\binom{q}{4}$ of such strings and $F_4^{(0)}=q^4-2q\binom{q}{3}+\binom{q}{4}$ which can again be verified to be bounded from above by $(0.99q)^4$ for any $q>2.$ Lastly, we consider $F_5^{(0)}.$ It is easy to see that if $q$-ary string of length $5$ is not in $F_5^{(0)},$ then it must have a strictly decreasing substring of length $3,$ which can start at either the first, second or third entry. Noting that the case that a string has two strictly decreasing substrings of length $3$ starting from both the first and third entry must also have a strictly decreasing substring of length $3$ from the second entry, by the Inclusion and Exclusion principle, we have
 \[F_5^{(0)}=q^5-3q^2\binom{q}{3}+2q\binom{q}{4}=\frac{7}{12}q^5+q^4-\frac{13}{12}q^3-\frac{q^2}{2}\le (0.99 q)^5.\]

 Now we claim that for any $m\ge 6,$ if $F_i^{(0)}<(0.99q)^i$ for $i=m-3,m-2,m-1,$ then we must also have $F_m^{(0)}<(0.99q)^m.$

 Note that for $\mathbf{v}=(v_1,\cdots, v_m)\in S_m^{(0)},$ it must satisfy one of the following conditions:
 \begin{itemize}
 \item $v_1\in\{0,1\}$ and $(v_2,\cdots, v_m)\in S_{m-1}^{(0)},$
 \item $v_1\ge 2, v_2\in\{0,v_1,v_1+1,\cdots, q-1\}$ and $(v_3,\cdots, v_m)\in S_{m-2}^{(0)}$ or
 \item $v_1\ge 2, v_2\in\{1,\cdots, v_1-1\}, v_3\ge v_2$ and $(v_4,\cdots, v_m)\in S_{m-3}^{(0)}.$
 \end{itemize}
 Hence,
 \begin{eqnarray*}
F_m^{(0)}&\le& 2F_{m-1}^{(0)}+\sum_{j=2}^{q-1}(q-j+1)F_{m-2}^{(0)}+\sum_{j=2}^{q-1}\sum_{i=1}^{j-1}(q-i)F_{m-3}^{(0)}\\
&=&2F_{m-1}^{(0)}+\frac{(q+1)(q-2)}{2}F_{m-2}^{(0)}+\frac{1}{3}q(q-1)(q-2)F_{m-3}^{(0)}\\
&<&2(0.99 q)^m-1+\frac{(q+1)(q-2)}{2} (0.99 q)^{m-2}+\frac{1}{3}q(q-1)(q-2)(0.99 q)^{m-3}\\
&=&0.99^{m-3}q^{m-2}\left(q^2\left(\frac{0.99}{2}+\frac{1}{3}\right)+q\left(2(0.99)^2-\frac{0.99}{2}-1\right)+\left(\frac{2}{3}-0.99\right)\right)\\
\end{eqnarray*}
which can be easily verified to be at most $(0.99)^m q^m,$ as claimed.
\end{proof}

Let $U_m$ be the set of all $q$-ary strings $\mathbf{v}=(v_1,\cdots, v_m)$ of length $m$ such that it contains $1\le i,j\le m-2$ where $v_i>v_{i+1}>v_{i+2}$ and $v_j\le v_{j+1}\le v_{j+2}$ and $G_m$ be its size. Hence we have $G_m\ge q^m - F_m^{(0)}-F_{m}^{(1)}> q^m (1-2(0.99)^m).$ We further define $R_n,$ the set of all $d$-regular $q$-ary strings of length $n.$ Then we have the following result.

\begin{lemma}\label{lem:QnlargeEff}
For a sufficiently large $n,$ there exists a positive integer $M>q^{n-1}$ and a one-to-one map $\Pi_q: \Sigma_q^{n-1}\rightarrow R_n$ which can be efficiently computed.
\end{lemma}

\begin{proof}
Set $m=\left\lfloor\frac{d}{2}\log n\right\rfloor-1$ and $\Delta=\left\lfloor\frac{n}{m}\right\rfloor.$ It is then easy to see that concatenating $\Delta~q$-ary strings taken from $U_m$ and padding it with an arbitrary $n-m\Delta q$-ary string produces an element of $R_n,$ which is a $d$-regular $q$-ary string of length $n.$ Hence the number of $d$-regular $q$-ary strings of length $n$ that can be constructed in this manner, which we denote by $M$ is

\[
M= G_m^\Delta q^{n-m\Delta}
> \left(q^m(1-2(0.99)^m)\right)^\Delta q^{n-m\Delta}
= q^{n}\left(1-2(0.99)^m\right)^\Delta.
\]
 Note that $2(0.99)^m=O\left(\frac{1}{2}^{\log n}\right)=O\left(\frac{1}{n}\right).$ On the other hand, we have $\Delta=O\left(\frac{n}{\log n}\right).$ This implies $2\Delta(0.99)^m = O\left(\frac{1}{\log n}\right).$ Hence for sufficiently large $n, 2\Delta(0.99)^m<1$ and we have

\begin{eqnarray*}
 (1-2(0.99)^m)^\Delta&=&1+ \sum_{j=1}^\Delta \binom{\Delta}{j}(-2(0.99)^m)^j
 > 1-\sum_{j=1}^\Delta\binom{\Delta}{j}(2(0.99)^m)^j
 > 1-\sum_{j=1}^\Delta \frac{(2\Delta(0.99)^m)^j}{j!}\\
 &>& 1- \sum_{j=1}^\Delta \frac{2\Delta(0.99)^m}{j!}
 = 1- 2\Delta(0.99)^m\sum_{j=1}^\Delta \frac{1}{j!}>1-2\Delta(0.99)^m (e-1).
 \end{eqnarray*}

We again note that $2(e-1)\Delta(0.99)^m= O\left(\frac{1}{\log n}\right).$ Hence for a sufficiently large $n,$ we have $1-2(e-1)\Delta(0.99)^m \ge \frac{1}{q}.$ This shows that for a sufficiently large $n, M > q^{n-1}.$ Hence $\Pi_q$ can be designed and efficiently computed in the same way as the one proposed in \cite[Lemma $5.12$]{GH21} with computation complexity polynomial in $n.$
\end{proof}

We conclude by summarizing our result in the following theorem.

\begin{theorem}\label{thm:ConsFromGH21}
Let $\mathcal{C}_0\subseteq \Sigma_2^n$ be the binary two-deletion correcting code constructed in \cite[Theorem 5.9]{GH21}. Then for any positive integer $q>2,$ there exists an explicit and efficiently encodable $2$ deletion correcting code $C_1\subseteq \Sigma_q^n$ with redundancy $5\log n+10\log\log n + 3\log q+O(1).$

Furthermore, utilizing our list-decodable code construction described in Corollary \ref{cor:GenConsLD} with $\mathcal{C}_0$ as the underlying binary $2$-deletion correcting code, there exists an explicit and efficiently encodable $q$-ary deletion code $C_2\subseteq \Sigma_q^n$ with redundancy $4\log n+10\log\log n+2\log q+O(1)$ that can be  list decoded against $2$ deletions with list size of at most $2.$

In both cases, the encoding complexity is $\mathrm{poly}(n).$ Furthermore, if the complexity of computing $\Pi_q^{-1}$ is $T_{\Pi_q,D}(n),$ then the decoding complexity is $2T_{\Pi_q,D}(n)+O(\log^2 q + n).$
\end{theorem}

\begin{rem}\label{rem:ConsFromGH21}{\rm
Note that the $q$-ary $2$-deletion correcting code $C_1$ defined in Theorem \ref{thm:ConsFromGH21} outperforms any of the existing construction of $2$-deletion codes for any $q.$ This includes the code constructed in \cite{SC22}.

Similarly, the $q$-ary $2$-deletion code $C_2$ defined in Theorem \ref{thm:ConsFromGH21} with list size at most $2$ outperforms our best construction from  Corollary \ref{cor:GenConsLDRed}.
}\end{rem}


\begin{thebibliography}{99}

\bibitem{Lev02} 
V. I. Levenshtein. Bounds for Deletion/Insertion Correcting Codes. In \textit{Proceedings IEEE International Symposium on Information Theory}, pp. 370. doi: \url{10.1109/ISIT.2002.1023642}, 2002.

\bibitem{YD04} 
 S. Yekhanin and I. Dumer, {\it Long Nonbinary Codes Exceeding the Gilbert-Varshamov
Bound for any Fixed Distance,} IEEE Trans. on Information Theory, vol.10(2004), 2357--2362.

\bibitem{VT65} 
 R. R. Varshamov and G. M. Tenengolts. Codes which Correct Single
Asymmetric Errors. In \textit{Automatika i Telemkhanika,} vol. 161, no. 3, pp. 288 -- 292. 1965 (in Russian).

\bibitem{Lev65} 
V. Levenshtein. Binary Codes Capable of Correcting Deletions, Insertions and Reversals. In~\textit{Soviet Physics-Doklady,} vol. 10, no. 8, pp. 707 -- 710. 1966. Translated from \textit{Doklady Akademii Nauk SSSR,} vol. 163, no. 4, pp. 845 -- 848. 1965 (in Russian).

\bibitem{Ten84} 
G. Tenengolts. Nonbinary Codes, Correcting Single Deletion or Insertion. In \textit{IEEE Transactions on Information Theory,} vol. 30, no. 5, pp. 766 -- 769. 1984.

\bibitem{HF02}
A. Helberg and H. Ferreira. On Multiple Insertion/Deletion Correcting Codes. In \textit{IEEE Transactions on Information Theory,} vol. 48, no. 1, pp. 305 -- 308. 2002.

\bibitem{LN16} 
T. Le and H. Nguyen. New Multiple Insertion/Deletion Correcting Codes for Non-binary Alphabets. In~\textit{IEEE Transactions on Information Theory,} vol. 62, no. 5, pp. 2682 -- 2693. 2016.

\bibitem{BGZ18} 
J. Brakensiek, V. Guruswami and S. Zbarsky. Efficient Low-Redundancy Codes for Correcting Multiple Deletions. In \textit{IEEE Transactions on Information Theory,} vol. 64, no. 5, pp. 3403 -- 3410. 2018.

\bibitem{SRB18}
 J. Sima, N. Raviv and J. Bruck. Two Deletion Correcting Codes from Indicator Vectors. In \textit{IEEE Transactions on Information Theory}, vol. 66 no. 4 pp. 2375--2391, doi: \url{10.1109/ISIT.2018.8437868}, 2019.

\bibitem{GS19} 
R. Gabrys and F. Sala. Codes Correcting Two Deletions. In \textit{IEEE Transactions on Information Theory}, vol. 65, no. 2, pp. 965--974, doi: \url{10.1109/TIT.2018.2876281}, 2019.

\bibitem{SB19} 
J. Sima and J. Bruck, Optimal k-Deletion Correcting Codes. In \textit{2019 IEEE International Symposium on Information Theory (ISIT)}, pp. 847--851, doi: \url{10.1109/ISIT.2019.8849750}, 2019.


\bibitem{GH21}
 V. Guruswami and J. H{\aa}stad. Explicit Two-Deletion Codes with Redundancy Matching the Existential Bound. In \textit{IEEE Transactions on Information Theory}, vol. 67, no. 10, pp. 6384--6394. doi: \url{10.1109/TIT.2021.3069446}, 2021.

\bibitem{SWWY20} I. Smagloy, L. Welter, A. Wachter-Zeh and E. Yaakobi, Single-Deletion Single-Substitution Correcting Codes, In \textit{IEEE International Symposium on Information Theory (ISIT)}, pp. 775--780, doi: {10.1109/ISIT44484.2020.9174213}, 2020.

\bibitem{SPCH22} W. Song, N. Polyanskii, K. Cai and X. He. Systematic Codes Correcting Multiple-Deletion and Multiple-Substitution Errors, in \textit{IEEE Transactions on Information Theory}, vol. 68, no. 10, pp. 6402--6416. doi: \url{10.1109/TIT.2022.3177169}, 2022.






\bibitem{Zhou20} 
Z. Zhou. 2-Deletion Codes: Beyond Binary. Master Thesis, Carnegie Mellon University, 2020. [Online] Available at: \url{http://reports-archive.adm.cs.cmu.edu/anon/2020/CMU-CS-20-103.pdf}



\bibitem{SC22} 
W. Song and K. Cai. Non-binary Two-Deletion Correcting Codes and Burst-Deletion Correcting Codes. \textit{arXiv:2210.14006v1}, 2022. [Online] Available at: \url{http://arxiv.org/abs/2210.14006v1}


\bibitem{Ton91} 
A. Tonelli, Bemerkung {\"{U}}ber die Aufl{\"{o}}sung Quadratischer Congruenzen. In \textit{G{\"{o}}ttinger Nachrichten,} pp. 344--346, 1891.

\bibitem{Sha73} 
D. Shanks. Five Number-Theoretic Algorithms. In \textit{Proceedings of the Second Manitoba Conference on Numerical Mathematics, Congressus Numerantium, Utilitas Mathematica,} no. 7, pp 51--70, 1973


\end{thebibliography}
\end{document}